\documentclass[reqno,a4paper,final]{amsart}
\renewcommand{\paragraph}{\subsubsection*}
\usepackage[UKenglish]{babel}
\usepackage{microtype}
\usepackage{booktabs}

\usepackage{lineno}
\usepackage{url}
\usepackage[
  pdftex,
  bookmarks=true,%                   
  bookmarksnumbered=true,%           
  hypertexnames=false,%              
  breaklinks=true,%                  
  linkbordercolor={0 0 1}%          
]{hyperref}

\providecommand{\urlstyle}[1]{}
\providecommand{\doi}[1]{\href{http://dx.doi.org/#1}{\nolinkurl{doi:#1}}}

\usepackage{hyperref}
\usepackage{float}
\usepackage{amsmath}
\usepackage{amsfonts}
\usepackage{amsthm}
\usepackage{mathtools}
\usepackage[makeroom]{cancel}
\usepackage{pgf}
\usepackage{pgfplots}
\usepackage{tikz}
\usepackage{microtype}
\usetikzlibrary{decorations.pathreplacing,calligraphy}
\usepackage[normalem]{ulem}

\usepackage{xcolor}
\definecolor{RoyalPurple}{rgb}{0.47, 0.32, 0.66}
\usepackage[nameinlink,capitalise]{cleveref}
\usepackage{caption}
\usepackage{subcaption}
\usepackage{thm-restate}

\newtheorem{definition}{Definition}
\newtheorem{example}[definition]{Example}
\newtheorem{theorem}{Theorem}
\newtheorem{proposition}[theorem]{Proposition}

\newtheorem{lemma}[theorem]{Lemma}

\newtheorem{claim}[theorem]{Claim}

\newcommand{\N}{{\mathbb{N}}}
\newcommand{\Z}{\mathbb{Z}}
\newcommand{\Q}{\mathbb{Q}}
\newcommand{\K}{\mathbb{K}}
\newcommand{\LL}{{\mathbb{L}}}
\newcommand{\OO}{{\mathcal{O}}}
\newcommand{\Spl}{{\mathrm{Spl}}}

\providecommand*{\eu}%
{\ensuremath{\mathrm{e}}}
\providecommand*{\iu}%
{\ensuremath{\mathrm{i}}}

\newcommand{\seq}[3][{}]{\langle #2 \rangle_{#3}^{#1}}

\usepackage{amssymb}

\usepackage{pgf}
\usepackage{pgfplots}

\begin{document}

\title[Hypergeometric Sequences with Quadratic Parameters]{The Membership Problem for Hypergeometric Sequences with Quadratic Parameters}

\author[G.~Kenison]{George Kenison}
\address{George Kenison, Institute of Logic and Computation, TU Wien, Vienna, Austria}
\email{george.kenison@tuwien.ac.at}

\author[K.~Nosan]{Klara Nosan}
\address{Klara Nosan, Universit\'e Paris Cité, CNRS, IRIF, Paris, France}
\email{nosan@irif.fr}

\author[M.~Shirmohammadi]{Mahsa Shirmohammadi}
\address{Mahsa Shirmohammadi, Universit\'e Paris Cité, CNRS, IRIF, Paris, France}
\email{mahsa@irif.fr}

\author[J.~Worrell]{James Worrell}
\address{James Worrell, Department of Computer Science, University of Oxford, Oxford, UK}
\email{jbw@cs.ox.ac.uk}

\thanks{
George Kenison gratefully acknowledges the support of ERC consolidator grant ARTIST
101002685 and WWTF grant ProbInG ICT19-018.
Klara Nosan and Mahsa Shirmohammadi are supported by International Emerging Actions grant (IEA'22), by ANR grant VeSyAM (ANR-22-CE48-0005) and  by the grant CyphAI  (ANR-CREST-JST). 
James Worrell is supported by EPSRC fellowship EP/X033813/1.
}

\maketitle

\DeclareRobustCommand{\gobblefive}[5]{}
\DeclareRobustCommand{\gobblenine}[9]{}
\newcommand*{\SkipTocEntry}{\addtocontents{toc}{\gobblenine}}

\begin{abstract}
Hypergeometric sequences are rational-valued sequences that satisfy first-order linear recurrence relations with polynomial coefficients; that is, a hypergeometric sequence
\(\seq[\infty]{u_n}{n=0}\) is one that satisfies a recurrence of the form
$f(n)u_n = g(n)u_{n-1}$ where $f,g \in \Z[x]$.

In this paper, we consider the Membership Problem for hypergeometric sequences: given a hypergeometric sequence \(\seq[\infty]{u_n}{n=0}\) and a target
value $t\in \Q$, determine whether $u_n=t$ for some index \(n\).
We establish decidability of the Membership Problem under the
assumption that either (i)~$f$ and $g$ have distinct splitting fields or (ii)~$f$ and $g$ are monic polynomials that both split over a quadratic extension of $\Q$.
Our results are based on an analysis of the prime divisors of polynomial sequences $\langle f(n) \rangle_{n=1}^\infty$
and $\langle g(n) \rangle_{n=1}^\infty$ appearing in the recurrence relation.
\end{abstract}

\maketitle

\section{Introduction}

\paragraph{Background and Motivation}
Recursively defined sequences are ubiquitous in mathematics and computer science.  A fundamental open problem in this context is the decidability of the \emph{Membership Problem}, which asks to determine whether a given value is an element of a given sequence.  The Skolem Problem for \emph{C-finite} sequences (those sequences that satisfy a linear recurrence relation with constant coefficients) is the best known variant of the Membership Problem.  The Skolem Problem asks to determine whether a given C-finite sequence vanishes at some index~\cite{everest2003recurrence}.  Decidability of this problem is known for recurrences of order at most four \cite{mignotte1984distance,vereshchagin1985occurence} but is open in general. Proving decidability of the Skolem Problem would be equivalent to giving an effective proof of the celebrated Skolem--Mahler--Lech Theorem, which states that every non-degenerate C-finite sequence that is not identically zero has a finite set of zeros.

In this paper we consider the most basic case of the Membership Problem for a class of \emph{P-finite} sequences (those sequences that satisfy a linear recurrence with polynomial coefficients).  Specifically, we consider the Membership Problem for the class of hypergeometric sequences.  A rational-valued sequence $\langle u_n \rangle_{n=0}^\infty$ is \emph{hypergeometric} if it satisfies a recurrence relation of the form
 \begin{equation}\label{eq:rel}
 f(n)u_{n} - g(n)u_{n-1} = 0 \, ,
\end{equation} 
 where $f,g \in \Z[x]$ are polynomials, and  $f(x)$
 has no non-negative integer zeros.   
 By the latter assumption on~$f(x)$, the  recurrence relation~\eqref{eq:rel} 
 uniquely defines an infinite sequence of rational numbers once the initial value $u_0\in\Q$ is specified.
The term 
 \emph{hypergeometric} was introduced by John Wallis in the 17th century~\cite{wallis1655arithmetica} and
 hypergeometric sequences and their 
 associated generating functions, the hypergeometric series, 
 have a long and illustrious history in the mathematics literature.  In particular, hypergeometric series encompass many of the common mathematical functions
 and have numerous applications in analytic combinatorics~\cite{FS09, Concrete11}.

The Membership Problem for hypergeometric sequences asks, given a recurrence~\eqref{eq:rel}, initial value $u_0\in \Q$, and target $t\in \Q$, whether $t$ lies in the sequence $\langle u_n \rangle_{n=0}^\infty$.  At first glance, this problem may seem easy to decide.  Without loss of generality we can assume that the sequence $\langle u_n \rangle_{n=0}^\infty$ either diverges to infinity or converges to a finite limit.
If the sequence does not converge to~$t$
then one can
compute a bound $B$ such that $u_n \neq t$ for all $n>B$.  Such a bound can also be computed in case one is promised that $\langle u_n \rangle_{n=0}^\infty$ converges to $t$, by using the fact that the convergence to $t$ is ultimately monotonic.  However the above case distinction does not suffice to show decidability of the Membership Problem!  The problem is that 
it is not known how to decide whether a hypergeometric seqeuence converges to a given rational limit.
The latter is related to deep conjectures about the gamma function (see the discussion below). In this paper we will take a different route to establish decidability of certain cases of the Membership Problem.

\paragraph{Contributions}
We approach the Membership Problem by considering the prime divisors of the values of a hypergeometric sequence
$\langle u_n \rangle_{n=0}^\infty$.  The overall strategy is to exhibit an effective threshold $B$ such that for all $n>B$ there is a prime divisor of $u_n$ that is not a divisor of the target $t$.  Our two main contributions are as follows:
\begin{itemize}
    \item The Membership Problem for hypergeometric sequences whose polynomial coefficients (as in \eqref{eq:rel}) have distinct splitting fields is decidable (\autoref{theo-distinctsplit}).
    \item The Membership Problem for hypergeometric sequences whose polynomial coefficients are monic and split over a quadratic field is decidable (\autoref{theo-decide-quad}).
\end{itemize}

The proofs of our main results involve two different implementations of our general strategy.  The proof of \Cref{theo-distinctsplit} applies the Chebotarev density theorem to find a single prime $p\in \Z$ that does not divide the target $t$ but divides all members of an infinite tail of the sequence.  Meanwhile, the proof of \Cref{theo-decide-quad} shows that for all sufficiently large $n$ there exists a prime $p$, that is allowed to depend on $n$, such that $p$ divides $u_n$ but not $t$.  To find such a prime we rely on (a mild generalisation of) a result of~\cite{everest07} concerning prime divisors of the values of a quadratic polynomial.

\Cref{theo-distinctsplit} expands the class of sequences for which the Membership Problem can be solved and further isolates its hard instances.  The paper~\cite{NPSW022} handles perhaps the easiest sub-case of the Membership Problem that does not fall under~\Cref{theo-distinctsplit}, namely when the polynomial coefficients both split over~$\Q$.  The second main result of the present paper handles another naturally occurring sub-case: when the polynomial coefficients split over the ring of integers of a quadratic field $\K$.  A common refinement of these two cases---that the polynomial coefficients split over $\K$---is the subject of current research.  Generalisations of the  results of~\cite{everest07} to higher-degree polynomials are a subject of ongoing research in number theory and potentially would allow us to extend our approach beyond the quadratic case.
 
\paragraph{Related Work}
There is a growing body of work that addresses membership and threshold problems for sequences satisfying low-order polynomial recurrences.  Here the \emph{Threshold Problem} asks to determine whether every term in a sequence lies above a given threshold, for example, whether every term is non-negative.  

The recent preprint~\cite{kenison2022applications} 
establishes decidability results (some conditional on Schanuel's Conjecture) for both the Membership and Threshold Problems for hypergeometric sequences.
The approach of~\cite{kenison2022applications}
relies on transcendence theory for the gamma function (as well as underlying properties of modular functions established by Nesterenko~\cite{nesterenko1996modular}).
By contrast, the algebraic techniques of the present paper seem appropriate only for the Membership Problem.  
We note that the approach of~\cite{kenison2022applications} requires certain restrictions, e.g., decidability is only unconditional when the parameters are drawn from imaginary quadratic fields.

The problem of deciding positivity of order-two P-finite sequences and of deciding the existence of zeros in such sequences is considered in~\cite{KauersP10,KenisonKLLMOW021,NeumannO021,PillweinS15}. These works all place syntactic restrictions on the degrees of the polynomial coefficients involved in the recurrences, and all four give algorithms that are not guaranteed to terminate for all initial values of a given recurrence.  For example, in~\cite{KauersP10}
the termination proof of the algorithm for determining positivity of order-two sequences requires 
that the characteristic roots of the recurrence be distinct and that one is working with a generic solution of the recurrence (in which the asymptotic rate of growth corresponds to the dominant characteristic root of the recurrence).   
Simple manipulations show that the Membership Problem considered in this paper is equivalent to the problem of finding a zero term in an order-two P-finite sequence $\langle u_n \rangle_{n=0}^\infty$ arising as a sum of two hypergeometric sequences.

Links between the Membership and Threshold Problems and the Rohrlich--Lang Conjecture appear in previous works \cite{kenison2020positivity, NPSW022}.  Here the Rohrlich--Lang Conjecture concerns multiplicative relations for the gamma function evaluated at rational points.

The $p$-adic techniques used in the present paper bear many similarities with 
work on developing criteria for hypergeometric sequences to be integer valued.
For example, work by Landau in 1900 \cite{landau1900factorielles} uses \(p\)-adic analysis to establish a necessary and sufficient condition for integrality in the so-called class of \emph{factorial} hypergeometric sequences.
In more recent work, Hong and Wang \cite{hongarxiv2016} establish a criterion for the integrality of hypergeometric series with parameters from quadratic fields.
We observe that some of the intermediate asymptotic results in Hong and Wang's note are close to \cite[Corollary~3.1]{Moll09}
(\autoref{prop:moll} herein).

\paragraph{Structure}

The remainder of this paper is structured as follows.
We briefly review preliminary material in \cref{sec-pre}, including some standard assumptions about instances of the Membership Problem that can be made without loss of generality.
In \cref{sec-polyseq}, we recall useful technical results on the prime divisors of hypergeometric sequences that satisfy monic recurrence relations (see \eqref{eq:poly}).
In \cref{sec-unequalsplitting}, we prove \cref{theo-distinctsplit}.
The proof of \cref{theo-decide-quad} is given in \cref{sec-quad}.
We discuss ideas for future research in \cref{sec:discussion}.
The remaining appendices prove technical results omitted from the main text.

\section{Preliminaries}
\label{sec-pre}

\paragraph{Hypergeometric Sequences}

A hypergeometric sequence $\langle u_n\rangle_{n=0}^\infty$ is a sequence of rational numbers that satisfies 
a recurrence of the form \eqref{eq:rel}
where $f,g \in \Z[x]$ are polynomials, and  $f(x)$
has no non-negative integer zeros. 
By the latter requirement on~$f(x)$, the  recurrence~\eqref{eq:rel} 
uniquely defines an infinite sequence of rational numbers once the initial element $u_0$ 
is specified.

An instance of the Membership Problem for hypergeometric sequences consists of a recurrence~\eqref{eq:rel}, an initial value 
$u_0 \in \Q$, and a target $t \in \Q$.  
The problem asks to decide whether there exists \(n\in\N\) such that \(u_n = t\).
We say that such an instance is in \emph{standard form} if~(S1) the initial condition is $u_0=1$; (S2)~the polynomial $g(x)$ has no positive integer root; (S3)~the target $t$ is non-zero;
(S4)~the polynomials $f$ and $g$ have the same degree and leading coefficient.

For the purposes of deciding the Membership Problem, we can assume without loss of generality that all instances are in standard form.  An arbitrary instance can be transformed into one satisfying Condition~(S1) by 
multiplying the sequence and target by a suitable constant.  Instances of the  Membership Problem that fail to satisfy 
Conditions~(S2) and (S3) are trivially solvable.  The positive integer roots of $g$ can be computed and for any such root $n_0$, we have $u_n=0$ for all $n\geq n_0$.  %
Finally,
for recurrences that fail Condition~(S4) we have that \[ \frac{u_n}{u_{n-1}}=\frac{g(n)}{f(n)} \] either converges to $0$ or diverges in absolute value.  Under the assumption that $t\neq 0$, in each case we can compute an effective threshold $n_0$ such that $u_n\neq t$ for all $n\geq n_0$.

\paragraph{The $p$-adic valuation}
Let $p\in \N$ be a  prime.
 Denote by~$v_p:\Q \to \Z \cup\{\infty\}$
the $p$-adic valuation on~$\Q$. 
Recall that for a  non-zero  number~$x\in \Q$, 
$v_p(x)$ is the unique integer such that~$x$ can be written in the form
\[x=p^{v_p(x)}\; \frac{a}{b}\]
where $a,b\in \Z$ and $p$ divides neither $a$ nor $b$.
The value $v_p(0)$ is defined to be $\infty$. 
The valuation possesses two important properties:
\begin{enumerate}
	\item[-]$v_p(x+y)\geq \min\{v_p(x),v_p(y)\}$ \, (\emph{strong triangle inequality}),
	\item[-]$v_p(xy)=v_p(x)+v_p(y)$ \, (\emph{multiplicative property}).
\end{enumerate}

\paragraph{Asymptotic estimates for series over primes}
Given  ${\sim} \in {\{<,=,> \}}$ and $x\in \Q$,
we denote sums over primes \(p\in\N\) such that 
\(p \sim x\) by \(\sum_{p \sim x}\).
Let \(\pi(x) := \sum_{p\le x} 1\) count the number of primes of size at most~\(x\).
The following result is a consequence of the celebrated Prime Number Theorem.
    \begin{theorem}\label{thm:pnt} For \(\pi(x)\) as above, we have
            \begin{equation*}
               \pi(x) = \frac{x}{\log x} + O\Bigl(\frac{x}{\log^2 x}\Bigr).
            \end{equation*}
    \end{theorem}

As an aside, an element \(a\in\Z\) is a \emph{square} modulo a  prime \(p\in \N\) if there exists an \(x\in\Z\) such that  \(x^2 \equiv a \pmod{p}\). 
An element \(a\in\Z\) is a \emph{quadratic residue} modulo  \(p\) if \(a\) is both a square modulo \(p\), and furthermore \(a\) and \(p\) are co-prime.
We denote by \(\mathcal{L}_p\) the set of quadratic residues modulo \(p\).

Recall the first of Mertens' three theorems \cite{mertens1874zahlentheorie}  (see also \cite[Theorem~4.10]{apostol1998introduction}),
\[
 \sum_{p \leq x } \frac{\log p}{p} =  \log x + O(1) \,.
 \] 
In the sequel we shall make use of the following refinement of Mertens' theorem.
\begin{proposition} \label{prop:apostol_primes_beta}
Suppose that \(a \in \Z\) is not a  perfect square. 
Then
\begin{equation*}
\sum_{p\leq x, \, a \in \mathcal{L}_p}  \frac{\log p}{p} =  \frac{1}{2}\log(x)+O(1).
\end{equation*}
\end{proposition}
\autoref{prop:apostol_primes_beta} appears in work by Selberg  \cite[Equation (3.3)]{selberg1950pnt-ap}
on an elementary proof of Dirichlet's theorem in arithmetic progressions.

\section{Monic Recurrences}
\label{sec-polyseq}
In this section, we study  
hypergeometric sequences~$\langle u_n \rangle_{n=0}^\infty$, satisfying  first-order recurrences of the special form
\begin{equation}\label{eq:poly}
u_n=f(n) u_{n-1} \quad \text{ and } \quad u_0=1,
\end{equation}
where $f\in \Z[x]$ has no non-negative integer roots.  We call such a
recurrence \emph{monic}.  We analyse the prime divisors of
sequences~$\langle u_n \rangle_{n=0}^\infty$ that satisfy
such a monic
recurrence.  In particular, we recall two results that will serve as
stepping stones toward our main decidability theorems in the
subsequent sections.  
Following~\cite{Moll09}, for a fixed prime \(p\), the first result establishes an asymptotic estimate 
for the $p$-adic valuation $v_p(u_n)$
as $n$ tends to infinity.  Next, following~\cite{everest07}, when
$f$ is a quadratic polynomial we prove a result that yields asymptotic
estimates on the size of the largest prime divisors of $u_n$ as $n$
tends to infinity.  The restriction on the degree is necessary given the
state of the art: estimates on large prime divisors constitute hard
open problems in the theory of
polynomials~\cite{hinz1996multiplicative,heathbrown2001largest}.

\subsection{Asymptotic growth of valuations}
\label{subsec:asymval}

Let $p\in \N$ be  prime.
Consider a hypergeometric sequence $\langle u_n \rangle_{n=0}^\infty$, satisfying a monic recurrence~\eqref{eq:poly}. 
Since $u_n=\prod_{k=1}^n f(k)$, we have 
\begin{equation*}
	v_p(u_n) = \sum_{k=1}^n v_p(f(k)). 
\end{equation*}
In this section we recall the result of~\cite{Moll09} that
characterises the asymptotic growth of $v_p(u_n)$ in terms of the
number of roots of $f$ in $\Z/p\Z$.  The key tool in this argument is
Hensel's Lemma.

\begin{theorem}[Hensel's Lemma {\cite[Theorem 4.7.2]{gouvea2020padic}}]
Let \(f(x)\in\Z[x]\)  and assume that there exist polynomials 
\(g(x)\) and \(h(x)\) such that: 
i) \(g(x)\) is monic, 
ii) \(g(x)\) and \(h(x)\) are relatively prime modulo \(p\), and 
iii) \(f(x) = g(x)h(x) \pmod{p}\).

Then for all $e>0$ there exist polynomials \(g_1(x),h_1(x)\in\Z[x]\) such that: i) \(g_1(x)\) is monic, ii) \(g_1(x)\equiv g(x) \pmod{p}\) and \(h_1(x)\equiv h(x) \pmod{p}\), and \(f(x)=g_1(x)h_1(x) \pmod{p^e}\).
\end{theorem}

Define 
a \emph{Hensel prime} for $f\in \Z[x]$ to be a prime that does not divide the
discriminant of any irreducible factor of $f$.  Since the discriminant
of an irreducible polynomial is non-zero, all but finitely many primes
are Hensel primes for a given polynomial.

Given a prime $p$, suppose that $f \in \Z[x]$ has $m$ roots in
$\Z/p\Z$,
i.e., suppose that $f$ factors as
\[ f=(x-\alpha_1)^{m_1} \cdots (x-\alpha_\ell)^{m_\ell} g(x) \pmod{p},\]
where $\alpha_1,\ldots,\alpha_\ell\in \Z$, $g\in \Z[x]$ has no root  modulo $p$, and $m=m_1+\cdots+m_\ell$.
In this case, if 
$p$ is a Hensel prime for $f$ then for all $e>0$ we can apply Hensel's Lemma to obtain a factorisation
\[f(x)=(x-\beta_1)^{m_1} \cdots (x-\beta_\ell)^{m_\ell} h(x) \pmod{p^e}\]
where $\beta_1,\ldots,\beta_\ell \in \Z$, and $h\in \Z[x]$ has no root
modulo $p$.  
In other words, $f$ has exactly $m$ roots in the ring $\Z/p^e\Z$.

The following result is a reformulation
of~\cite[Corollary 3.1]{Moll09}.  For later use, we formulate the result so as to make explicit the 
dependence of the bounds for $v_p(u_n)$ on the prime $p$.  The proof remains the same.
\begin{proposition}[{\cite[Corollary~3.1]{Moll09}}]
\label{prop:moll}
Suppose that $\langle u_n \rangle_{n=0}^\infty$ satisfies the monic recurrence in Equation~\eqref{eq:poly} with polynomial coefficient
$f \in \Z[x]$.  
Let $p$ be a Hensel prime of $f$ such that $f$ has $m$
roots modulo $p$.
Then there exist  effectively computable constants~$\varepsilon, n_0>0$ such that if $n>n_0$, 
\[m\Big(\frac{n}{p-1}- \frac{\varepsilon \log n}{\log p}\Big) \leq v_p(u_n) \leq m\Big(\frac{n}{p-1}+ \frac{\varepsilon \log n}{\log p}\Big)\]
where $\varepsilon$ depends only on~$f$.
\end{proposition}
\begin{proof}
The function~$|f(x)|$ is eventually monotonically increasing on~$\N$.
There exists an effectively computable bound~$n_0$ such that 
for all $n \geq n_0$ and all $1 \leq k \leq n$, the inequality 
$|f(k)| \leq |f(n)|$ holds. 

Furthermore, there exists an effective constant $\varepsilon_0>0$,
independent of $p$, such that for all $n\geq n_0$
and all $1\leq k \leq n$ we have 
\[ |f(k)|< n^{\varepsilon_0} = p^{\varepsilon_0 \log n / \log p}.\]

Fix $n\geq n_0$ and define $e_{\max}$ to be the smallest power of~$p$ such that $p^{e_{\max}-1} \leq  |f(n)| < p^{e_{\max}}$.
Then
\begin{equation}
 \label{eq:emax}
e_{\max} \leq \frac{\varepsilon_0 \log n}{\log p}. 
\end{equation}

Since $p$ is a Hensel prime, by Hensel's Lemma, there is a factorisation 
\[  f(x)=(x-\beta_1)^{m_1} \cdots (x-\beta_\ell)^{m_\ell} h(x) 
\pmod{p^{e_{\max}}}. \]
where $m=m_1+\cdots+m_\ell$ and $h$ has no zero modulo~$p$.  

Denote by $\mathbb{I}\{p^e \mid x\}$ the function such that
\[ \mathbb{I}\{p^e \mid x \}:=
\begin{cases}
    1 & \text{ if }\ p^e \mid x, \\
    0 & \text{ otherwise.}
\end{cases}
\]
\color{black}

Since $v_p(f(k)) \leq e_{\max}$ for all~$k\leq n$, we have

\begin{align}
v_p(u_n) =& \sum_{k=1}^n v_p(f(k))\notag \\
           =& \sum_{k=1}^n \sum_{i=1}^\ell m_i \, v_p(k-\beta_i) \notag \\
        =& \sum_{k=1}^n \sum_{i=1}^\ell \sum_{e=1}^{e_{\max}} m_i \, \mathbb{I}\{p^e \mid k - \beta_i\} \notag\\
        =& \sum_{e=1}^{e_{\max}}\sum_{i=1}^\ell
        \sum_{k=1}^n
        m_i \mathbb{I}\{p^e \mid k - \beta_i\}. \label{eq:TAG}
\end{align}

Now for all $1\leq e\leq e^{\max}$ the set $\{ k \in \N : p^e \mid k-\beta_i\}$
is an arithmetic progression with common difference $p^e$ and so
\begin{equation}
\label{eq-m-AP}
		\ \frac{n}{p^e} -1  \leq  \sum_{k=1}^n \mathbb{I}\{p^e \mid k-\beta_i \} \leq \frac{n}{p^e} +1 ,
	\end{equation} 
Combining inequality~\eqref{eq-m-AP} with Equation~\eqref{eq:TAG} we obtain
\begin{equation}
 	\label{eq-sum-III}
	m \sum_{e=1}^{e_{\max}} \Big( \frac{n}{p^e} -1 \Big)  \leq  v_p(u_n)  \leq m \sum_{e=1}^{e_{\max}} \Big( \frac{n}{p^e} +1 \Big).
\end{equation}

Let $\varepsilon := \varepsilon_0+1$.
The desired result follows by sandwiching the term 
$\sum_{e=1}^{e_{\max}} \frac{1}{p^e}$ in~\eqref{eq-sum-III} by
\[
\frac{1-|f(n)|^{-1}}{p-1}\leq
\frac{1-p^{-e_{\max}}}{p-1}  =\sum_{e=1}^{e_{\max}}  \frac{1}{p^e} \leq \, \frac{1}{p-1}\]
in combination with the upper bound on $e_{\max}$ in~\eqref{eq:emax}.
\end{proof}

\subsection{Asymptotic estimate for the largest prime divisor}
Fix a polynomial $f(x):=x^2+\beta \in\Z[x]$.  We assume that $-\beta$
is not a perfect square, which is equivalent to assuming that $f$
is irreducible.  Let $a,b\in \Q$ be such that $0 \leq a <b$.  Let
$c,d \in \N$.  For all $n\in \N$ we define
\[I(n):=\{k \in \N : an  \leq  k\leq bn\} \cap (c\N+d) \] and
\[F_n:= \prod_{k\in I(n)} f(k).\]

Informally speaking, the following theorem gives effective
super-linear lower bounds on the growth of the function that maps $n$
to the greatest prime divisor of $F_n$.  The result itself and the
proof are a slight generalisation of~\cite[Theorem 5.1]{everest07}.  The
main difference is that we permit $I(n)$ to be the intersection of an
interval and an arithmetic progression, whereas 
the work cited above considers unrefined intervals~$I(n)=\{1,\ldots,n\}$.

\begin{restatable}{theorem}{theobigprimespolyT} 
\label{theo:bigprimespolyT}
Let $M \in \N$. There exists an effectively computable bound~$B\in \N$
such that for all \(n>B\)  there exists a prime~$p>Mn$ that divides $F_n$.
\end{restatable}
\begin{proof}
Given $n\in \N$, we have the prime factorisation $F_n = \prod_{p} p^{e_p}$ 
where $e_p:= v_p(F_n)$ for each prime $p$.  Note that $e_p=0$ for
all but finitely many $p$.
Taking logarithms, we get
\begin{equation*}
	\log(F_n) =\sum_{p} e_p \log p . 
\end{equation*}
Partitioning the above sum into a sub-sum over primes at most $Mn$ and a sub-sum over primes greater than $Mn$, we obtain
\begin{equation}
\label{eq:logunn}
	\sum_{p >Mn}  e_p \log p = \log(F_n) - 	\sum_{p \leq Mn}   e_p
        \log p .
\end{equation} 

The theorem at hand follows from a lower bound on the sum
$\sum_{p >Mn} e_p \log p$ on the left-hand side of~\eqref{eq:logunn}.
To this end we have two sub-goals: give a lower bound on $\log(F_n)$
and an upper bound on $\sum_{p \leq Mn} e_p \log p$.

Write $A:=\frac{b-a}{c}$.
The following lower bound on $\log(F_n)$ is a consequence of Stirling's formula.
The proof is in Appendix~\ref{sec:app-lowerbownun}.
\begin{claim} 
  \label{claim:sizefunnyQ} We have the bound
  $\log(F_n) \geq 2A  (n\log n - n)$.
\end{claim}

The next task is give an upper bound on $\sum_{p \leq Mn}   e_p \log p$.
Here we follow the approach in~\cite{everest07} and further partition the sum into 
those primes $p<n$ (treated in~\Cref{claim-small-prime-F}) and those primes 
\(n\le p\le Mn\) (treated in~\Cref{claim-average-prime}).

\begin{claim}
\label{claim-small-prime-F}
There exist positive constants \(\varepsilon, n_0>0\) such that if \(n>n_0\), then
\begin{equation*}
  \sum_{p <n}  e_p \log p \leq  An\log n  + \varepsilon n. 
\end{equation*}
\end{claim}
\begin{proof}
Let $S_n$ be the set of primes $p<n$ such that 
$p$ divides $F_n$ and $p$ is a Hensel prime for $f$.
Observe that 
\[ \sum_{p < n}  e_p \log p - \sum_{p\in S_n}
e_p \log p \leq \varepsilon_0 \log n\]
for an effective constant $\varepsilon_0$.
Indeed, if $p<n$ is a prime divisor of ${F_n}$ 
that does not lie in $S_n$ then $p$ divides the discriminant 
of $f$---and there are finitely many such primes.
Thus to prove the claim it will suffice to show the following bound for some effective constant $\varepsilon_1$:
\begin{equation} \sum_{p \in S_n}  e_p \log p \leq
  An\log n  + \varepsilon_1  n. 
\label{eq:SUM2}
\end{equation}

For \(p\in S_n\), we establish an upper bound on $e_p$ which follows from
\Cref{prop:moll}:
\begin{equation}
  e_p \leq  \frac{2An}{p-1}+  \frac{\varepsilon_2 \log n}{\log p}.
\label{eq:BOUND}
  \end{equation}
  Here the constant $\varepsilon_2$ is effective
and independent of the prime $p$.
The justification is given in~\Cref{sec:app-lowerbownun}.

We next argue that there exist 
effective constants $\varepsilon_3,\varepsilon_4,n_1>0$  such that the following chain of inequalities is valid 
 for all $n\geq n_1$.  
 We have that
 \begin{eqnarray*}
 	\sum_{p\in S_n} e_p \log p & \leq & \sum_{p \in S_n} \left( \frac{2An }{p-1}+\varepsilon_2 \, \frac{\log n}{\log p} \right) \log p
  \qquad\mbox{(by~\eqref{eq:BOUND})}\\
 	& \leq & \, 2An \sum_{p \in S_n} \frac{\log p }{p-1} + \varepsilon_2 \pi(n) \log n\\
 	&\leq &   2An \sum_{p \in S_n} \frac{\log p }{p-1} + \varepsilon_3 n
                \qquad \mbox{(by \autoref{thm:pnt})}\\
   &=& 2An \sum_{p \in S_n} \frac{\log p }{p} \left(1+\frac{1}{p-1}\right)+ \varepsilon_3 n\\
 	& \leq  & 2An \sum_{p \in S_n} \frac{\log p }{p}+ \varepsilon_4.
  \end{eqnarray*} 

  No prime in $S_n$ divides the discriminant of $f$.  Since the latter
  is equal to $-4\beta$, no prime in $S_n$ divides $\beta$.
  In addition, every prime in $S_n$ is a divisor of $F_n$; i.e., a
  divisor of $k^2+\beta$ for some $k\in I(n)$, we have that $\beta$ is
  a quadratic residue modulo $p$ for every prime $p \in S_n$.
 Thus, for sufficiently large \(n\), we have that
\[ \sum_{p \in S_n} \frac{\log p}{p} \leq \frac{1}{2} \log n + \varepsilon_5\]
(by \Cref{prop:apostol_primes_beta})
for some effective constant $\varepsilon_5$.

The desired bound~\eqref{eq:SUM2} follows by combining the previous two inequalities and fixing \(\varepsilon_1 \ge 2A\varepsilon_5 + \varepsilon_4\).
\end{proof}

\begin{claim}
\label{claim-average-prime}
	There exist effectively computable constants \(n_0,\varepsilon>0\) 
	such that if \(n>n_0\),
then
\begin{equation*} 
\sum_{n\leq p \leq Mn}  e_p \log p \leq \varepsilon n.
\end{equation*}
\end{claim}

\begin{proof}
  Let $n \in \N$.  Suppose that \(p> (b-a)n\) is a prime divisor of
  \(F_n\).  For such primes, we shall first show that
  $e_p := v_p(F_n) \le 2$.  Assume, for a contradiction, that there
  are distinct integers $k_1 < k_2 < k_3$ in $I(n)$ such that $p$
  divides $k_1^2+\beta$, $k_2^2+\beta$, and $k_3^2+\beta$.  Then
  $p \mid k_1^2 - k_2^2$.  
  Since \(p\) is prime, either
  $p \mid k_1-k_2$ or \(p \mid k_1+k_2\).  Since
  \(0< k_2-k_1 < (b-a)n \le p\), we deduce that $p \mid k_1+k_2$.  By
  symmetric reasoning we have that $p \mid k_2+k_3$.  Thus $p$ must
  also divide $(k_2+k_3)-(k_1+k_2) = k_3 - k_1$.  However, this leads
  to a contradiction since $p \geq (b-a)n \geq k_3-k_1$.  Hence for
  each prime divisor \(p\mid F_n\) with \(p\ge (b-a)n\), we find that
  $e_p = v_p(F_n) \le 2$.

Thus we bound the summation in the statement of the claim by
\begin{equation*} 
    \sum_{n< p \leq Mn}  e_p \log p
    \le \sum_{p\le Mn} 2 \log p
    \le 2 \log(Mn) \pi(Mn).
\end{equation*}
The desired result follows from the estimate on \(\pi(x)\) given by the Prime Number Theorem  (\autoref{thm:pnt}). 
\end{proof}

We return to the proof of \autoref{theo:bigprimespolyT}.  
From
Equation~\eqref{eq:logunn}, \cref{claim-small-prime-F}, and \cref{claim-average-prime}, there exist positive constants
\(\varepsilon,n_0>0\) such that if \(n>n_0\) then
\begin{equation*}
\sum_{p > Mn}  e_p \log p \geq An\log n - \varepsilon n.
\end{equation*}

In turn, the above lower bound entails that for sufficiently
large~\(n\), there exist prime divisors \(p \mid F_n\) such that \(p>
Mn\).  This concludes the proof.
\end{proof}

\section{Decidability: different splitting fields}
\label{sec-unequalsplitting}
In this section we show decidability of the Membership Problem for recurrence sequences that satisfy a first-order relation of the form \eqref{eq:rel} subject to the condition that the polynomial coefficients 
$f,g\in \Z[x]$ have different splitting fields.
To this end, it is useful to introduce the following terminology.
Let $p$ be a Hensel prime for $fg$.
We say that the recurrence~\eqref{eq:rel} is 
\emph{$p$-symmetric} if 
the two polynomials $f$ and $g$ have the same number of roots 
in $\Z/p\Z$.  Otherwise we say that the recurrence is 
\emph{$p$-asymmetric}.

We first show decidability of the Membership Problem in the case
of $p$-asymmetric recurrences and then we apply the Chebotarev Density Theorem to show that every recurrence in which $f$ and $g$ 
have different splitting fields is $p$-asymmetric for infinitely many primes $p$.

\begin{lemma}\label{lem-p-assym}
There is a procedure to decide the Membership Problem 
for the class of hypergeometric sequences whose defining recurrences 
are $p$-asymmetric for some prime $p$.
\end{lemma}

\begin{proof}
Suppose that the hypergeometric sequence 
$\langle u_n \rangle_{n=0}^\infty$ satisfies 
the recurrence~\eqref{eq:rel} and moreover that
there is a prime 
$p$ with respect to which the recurrence is $p$-asymmetric.
We want to decide whether such a sequence reaches a given target value $t$.

Consider the sequences $\langle x_n \rangle_{n=0}^\infty$ and
$\langle y_n \rangle_{n=0}^\infty$ respectively defined by the monic
recurrences
$x_n=g(n)x_{n-1}$, $y_n=f(n)y_{n-1}$, with $x_0=y_0=1$.
Then $u_n = \frac{x_n}{y_n}$ and hence, for the aforementioned prime \(p\),
\[v_p(u_n)=v_p(x_n)-v_p(y_n) = \sum_{\ell=1}^n (v_p(g(\ell)) - v_p(f(\ell)))\]
by the multiplicative property.

Recall that $p$ is, by definition, a Hensel prime for both $f$ and $g$.
Hence,
by \Cref{prop:moll}, we obtain an asymptotic estimate of the form
\[|v_p(x_n)-v_p(y_n)| = \frac{|m_g-m_f|n}{p-1}  + O(\log n)\]
where \(m_f\) is the number of roots of \(f\) modulo \(p\) and \(m_g\) is defined similarly.
Here the implied constant depends on \(fg\) and \(p\).
The proof concludes by noting that $v_p(t)$ is a constant, whereas \(v_p(u_n)\) is bounded away from \(v_p(t)\) for sufficiently large \(n\) (note this threshold is computable).
We deduce that \(u_n\neq t\), again, for sufficiently large \(n\), from which the desired result follows.
\end{proof}

We now give a sufficient condition for a recurrence to be 
$p$-asymmetric.  We use the following 
consequence of the  Chebotarev Density Theorem.  
Let $\K$ be a Galois field of degree $d$ over $\Q$, and denote by $\OO$ its ring of integers. 
Let  $\Spl(\K)$ be the set of rational primes~$p$
such that  the ideal $p\OO$ 
totally splits in $\OO$, i.e., such that 
\(p\OO = \mathfrak{p}_1 \cdots \mathfrak{p}_d\)
where the $\mathfrak{p}_i$ are distinct prime ideals.
The following result appears as \cite[Corollary 8.39]{milneANT} and \cite[Corollary 13.10]{neukirch1999algebraic}.
	The latter reference attributes the result to Bauer.
\begin{theorem}
\label{theo-split-prime}
Let $\K$ and $\LL$ be  Galois extensions of~$\Q$ such that
$\K \neq \LL$. Then $\Spl(\K)$ and $\Spl(\LL)$ differ in infinitely many primes. 
\end{theorem}

We state the main theorem of this section.
\begin{restatable}{theorem}{theodistinctsplit} 
\label{theo-distinctsplit}
	There is a procedure to decide the Membership Problem 
for the class of hypergeometric recurrences~\eqref{eq:rel} whose polynomial coefficients have different splitting fields.
\end{restatable}
\begin{proof}[Proof of \autoref{theo-distinctsplit}]
Let $\langle u_n\rangle_{n=0}^\infty$ satisfy a recurrence~\eqref{eq:rel} for which the coefficients $f$ and $g$ have respective splitting fields $\K$ and $\LL$, with $\K\neq \LL$.
Recall that there are only finitely many primes that are not Hensel
primes for~$fg$. 
By Theorem~\ref{theo-split-prime}, there exists a Hensel prime for $fg$ that lies in exacly one of the two sets 
$\Spl(\K)$ and $\Spl(\LL)$.
For such a prime $p$, the recurrence~\eqref{eq:rel} is $p$-asymmetric.  Hence the result follows 
from~\Cref{lem-p-assym}.
\end{proof}

We note that the recurrence~\eqref{eq:rel}
can be $p$-asymmetric even when $f$ and $g$ have the same splitting field.
We demonstrate this phenomenon with the following example.

\begin{example}
Let $\langle u_n \rangle_{n=0}^\infty$ be the hypergeometric sequence defined by
\begin{equation}
\label{exam-rec-1}
f(n)u_{n} - g(n)u_{n-1} = 0 \quad \text{and} \quad u_0=1,
 \end{equation}
where
	\begin{equation*}
 f(x) :=  (x^2+1)(x^2-2) \quad \text{ and } \quad
 g(x) :=  x^4-2x^2+9.
 \end{equation*}

It is easily checked that both $f$ and $g$ have
splitting field $\Q(\sqrt{2},\iu)$.  However 
we show that $f$ and $g$ have different numbers of roots in~$\Z/7\Z$, i.e., 
the recurrence~\eqref{exam-rec-1} is $7$-asymmetric. 
	   
It is straightforward to verify that $7$ is a Hensel prime for~$fg$ by noting that it does not divide the discriminants of the respective irreducible factors of $f$ and $g$.
To show that the recurrence is $7$-asymmetric, observe first that $f$ factors as $(x+4)(x+3)(x^2+1)$
over $\Z/7\Z$, where~$x^2+1$ is irreducible; thus $f$ has two roots in $\Z/7\Z$.
On the other hand, \(g\) 
factors into a 
pair of irreducible quadratic polynomials over 
$\Z/7\Z$ and hence has no roots.

We can now follow the argumentation of 
\autoref{lem-p-assym} to
decide the Membership Problem for 
$\langle u_n\rangle_{n=0}^\infty$ with respect to any given target~$t \in \Q$.
\color{black}
Consider the monic recurrences  \(x_n = g(n)x_{n-1}\)
and \(y_n = f(n)y_{n-1}\), with initial conditions $x_0=y_0=1$.  Note that $v_7(u_n)=v_7(y_n)-v_7(x_n)$. 
Since $g$ has no roots in $\Z/7\Z$,  $v_p(g(k))=0$
for all integers $k> 0$. It follows that 
\(v_7(x_n) = \sum_{k=1}^n v_7(g(k)) = 0  \)
and hence that $v_7(u_n)=v_7(y_n)$.

To obtain bounds on  $v_7(y_n)$, note that \(|f(k)|\le n^4\) for all \(n\ge 2\) and \(1\le k\le n\).
\autoref{prop:moll} gives the inequality
    \begin{equation*}
        \frac{2n}{6}  - \frac{10\log n}{\log 7} \le  v_7(y_n) . 
    \end{equation*}
    For any target $t\in \Q$, the above bound allows us to compute a threshold $B$ such that for all $n>B$ we have 
Since \(v_7(u_n)=v_7(y_n)>v_7(t)\) and hence $u_n\neq t$.
\end{example}

\section{Decidability: quadratic splitting fields}
\label{sec-quad}

In this section, we focus on the  decidability of the Membership Problem
for recurrences 
\begin{equation}
f(n)u_n - g(n)u_{n-1}=0, \qquad u_0=1
\tag{\ref{eq:rel}}
\end{equation}
in which both $f,g\in \Z[x]$ are monic 
and split completely over a quadratic (degree-two) extension $\K$ of $\Q$.

Recall that a number  field \(\K\) is \emph{quadratic} if and only if there is a square-free integer \(\beta\) such that \(\K = \Q(\sqrt{\beta})\).
The assumption that \(f\) and \(g\) are both monic ensures that the roots of both polynomials are algebraic integers 
in \(\Q(sqrt{\beta}) \).
As shown in~\cite[Chapter 3]{stewart2016algebraic}, the following holds.

\begin{theorem} \label{thm:quadratic}
Suppose that $\beta\in\Z$  is  square-free.
Then the ring of algebraic integers in $\Q(\sqrt{\beta})$ has the form $\Z[\theta]$, where
\[\theta= \begin{cases}
	\sqrt{\beta} & \text{if }\beta \not\equiv 1 \pmod{4},\\[3pt]
	 \textstyle\frac{\sqrt{\beta}-1}{2} & \text{if } \beta\equiv 1 \pmod{4}.
\end{cases}
\]
\end{theorem}

The main result of the section is as follows.

\begin{theorem}
\label{theo-decide-quad}
The Membership Problem for recurrences of the
form~\eqref{eq:rel} is decidable under the assumption that $f,g$
are both monic and both split over a  quadratic extension $\K$ of
$\Q$.
\end{theorem}

The proof of Theorem~\ref{theo-decide-quad} is given in \Cref{subsec-partition-roots,subsec-thre,subsec-primediv,subsec-conclud}.
The details differ slightly according to the two cases
for the generator $\theta$ of the ring of integers of $\K$, as presented in
Theorem~\ref{thm:quadratic}.  
In the subsections below, we treat the case for
$\theta=\frac{\sqrt{\beta}-1}{2}$.  
The necessary adjustments for the case
$\theta=\sqrt{\beta}$ are given in Appendix~\ref{app-sec-quad}.
Henceforth we assume a normalised instance of the Membership Problem,
given by the recurrence~\eqref{eq:rel} and target~$t\in \Q$.
Our goal is to exhibit an effective bound~$B$ such that $u_n\neq t$
for all $n>B$.  
To this end, our strategy is to find $B$ such that for
all~$n>B$ there exists a prime that divides $u_n$ but not~$t$.
At the conclusion of the proof of \autoref{theo-decide-quad}, we demonstrate the argument and techniques with a worked example, namely \autoref{ex:example2} in \Cref{subsec-conclud}.

Let $\beta\equiv 1\pmod{4}$ be a square-free integer and
$\K=\Q(\sqrt{\beta})$ a quadratic field over which the polynomials $f$
and $g$ in~\eqref{eq:rel} split completely.  Let
$\theta:=\frac{\sqrt{\beta}-1}{2}$ be such that $\Z[\theta]$ is the
ring of integers of $\K$.  Write
$m_{\theta}(x):=x^2+x+\frac{1-\beta}{4} \in \Z[x]$ for the minimal
polynomial of $\theta$.

\subsection{Partitioning the roots of \texorpdfstring{$fg$}{fg}}
\label{subsec-partition-roots}
Let $\mathcal R$ be the set of roots of $fg$.
We partition \(\mathcal{R}\) into disjoint subsets (which we shall call the \emph{classes} of \(\mathcal{R}\)) with \(\alpha,\tilde{\alpha}\in \mathcal{R}\) in the same class
if and only if \(\alpha-\tilde{\alpha}\in\Z\).
We say that a subset of $\mathcal{S}\subseteq \mathcal {R}$  is \emph{balanced} if $f$ and
$g$ have the same number of roots in $\mathcal S$, counting repeated roots
according  to their multiplicity. 
A subset is \emph{unbalanced} otherwise.
The linchpin of the proof of \autoref{theo-decide-quad} is the balance of roots in the classes.

If each class (as above) is
balanced then the roots of $f$ and $g$ can be placed in a bijection
under which corresponding roots differ by an integer and have the same
multiplicity in $f$ and $g$ respectively.  
In this case, by cancelling
common factors in the expression
$u_n = \prod_{k=1}^n
\frac{g({k})}{f({k})}$, we see that
for $n$ sufficiently large $u_n$ is a rational function in $n$.  
For such an instance, 
the Membership Problem reduces to the problem of deciding whether a univariate polynomial with rational integer
coefficients has a positive integer root, which is straightforwardly
decidable. 
A detailed account for this argument is given in~\cite[Appendix B]{NPSW022}.

Let us now consider the case where there is an unbalanced class~\(\mathcal{C}\).
By the assumption that $f$ and $g$ have the same
degree, there must, in fact, be at least two unbalanced
classes.
It follows that there is an unbalanced class that is not
contained in $\Z$ (i.e., an unbalanced class of quadratic integers).

Here it is convenient to define the following linear
ordering on~$\mathcal R$.  
Given elements $a\theta+b$ and
$a'\theta+b'$ in $\mathcal R$
(where \(a,a',b,b'\in\Z\)),
define $a\theta+b \prec a'\theta+b'$ if
and only if one of the following four mutually exclusive conditions
holds:
\begin{enumerate}
\item $a' \leq 0 < a$,
\item $0 < a < a'$,
\item $a<a' \leq 0$,
\item $a=a'$ and $b<b'$.
\end{enumerate}

Note that the classes in $\mathcal R$ are intervals
with respect to the order~$\prec$.  Thus the order lifts naturally to a linear order on classes.
In particular, the \emph{least
  unbalanced class}~$\mathcal C_0$ is well-defined.  Let
$\alpha_0=a_0\theta+b_0$ be the greatest element in $\mathcal C_0$.
Then $ \{\alpha\in\mathcal{R}: \alpha \preccurlyeq \alpha_0 \} $ is
unbalanced 
because this set is a disjoint union of balanced classes and~\(\mathcal{C}_0\).
Further, $a_0> 0$
because the least unbalanced class is necessarily a subset of quadratic integers of the form \(a_0\theta + \Z\).
Here we note that the image of an unbalanced class
under the
automorphism of~$\K$ that interchanges $\sqrt{\beta}$ and
$-\sqrt{\beta}$ is likewise an unbalanced class and so $a_0>0$.

\colorlet{cyan}{purple}
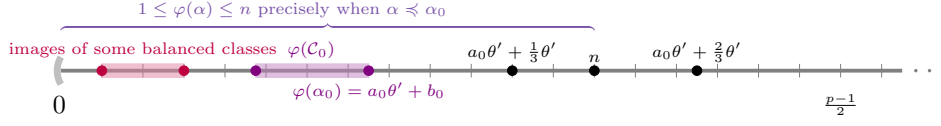
\begin{figure*}[t]
    \centering
    \begin{tikzpicture}
    \scalebox{.9}{
\def \x {3/5}

\draw[gray, ultra thick] (0,4) -- (20.5*\x,4);  
\foreach \y in {1,...,20}
\draw[draw=gray] (\y*\x,4.1)--(\y*\x,3.9);
\draw[gray, line width=1mm, opacity=.5] (0,4.2) arc[start angle=150, end angle=210, radius=.4cm];
\node[left] at (.2,3.5) {{\large  $0$}};
\node at (19*\x,3.5) {{\scriptsize $\frac{p-1}{2}$}};
\filldraw[gray] (20.5*\x+0.2,4) circle (0.5pt);
\filldraw[gray] (20.5*\x+0.4,4) circle (0.5pt);
\filldraw[gray] (20.5*\x+0.6,4) circle (0.5pt);

\filldraw [cyan!50!white, opacity=0.5] (1*\x,4.1) rectangle (3*\x,3.9);
\node at (2*\x,4.3) {\textcolor{cyan}{{\scriptsize images of some balanced classes}}};
 \node[left] at (1*\x+.4,3.7) {};%
 \filldraw[cyan] (1*\x,4) circle (2pt);
 \node[left] at (3*\x+.4,3.7) {};%
 \filldraw[cyan] (3*\x,4) circle (2pt);

 \filldraw [violet!50!white, opacity=0.5] (4.75*\x,4.1) rectangle (7.5*\x,3.9);
\node at (6.1*\x,4.3) {\textcolor{violet}{{\scriptsize $\varphi(\mathcal{C}_0)$}}};
 \node[left] at (4.75*\x+.4,3.7) {};%
 \filldraw[violet] (4.75*\x,4) circle (2pt);
 \node at (7.5*\x,3.7) {{\textcolor{violet}{\scriptsize $\varphi(\alpha_0) = a_0\theta' + b_0$}}};
 \filldraw[violet] (7.5*\x,4) circle (2pt);

\node[above] at (11*\x,4) {{\scriptsize $a_0 \theta' + \frac{1}{3}\theta'$}};
\filldraw[black] (11*\x,4) circle (2pt);

\node[above] at (13*\x,4) {{\scriptsize $n$}};
\filldraw[black] (13*\x,4) circle (2pt);

\node[above] at (15.5*\x,4) {{\scriptsize $a_0 \theta' + \frac{2}{3}\theta'$}};
\filldraw[black] (15.5*\x,4) circle (2pt);

 \draw [pen colour={RoyalPurple}, decorate, decoration = {calligraphic brace}, thick] (0,4.6) --  (13*\x,4.6);
 \node[right] at (1,4.9) {{\textcolor{RoyalPurple}{\scriptsize  $1 \leq \varphi(\alpha)  \leq n$ precisely when  $\alpha \preccurlyeq \alpha_0$}}};

}
\end{tikzpicture}
\caption{
Image of \(\varphi\) 
on \(\Z\) 
as well as  the positions of constants used in the proof of  \autoref{thm:quadratic} to determine that \(v_p(u_k)\neq 0\) for \(k\) that satisfy \(a_0\theta'+\frac{1}{3}\theta' \le k \le a_0\theta'+\frac{2}{3}\theta'\). 
Note that the preimages \(\alpha\in\mathcal{R}\) such that \(1\le \varphi(\alpha)\le n\) are precisely those roots for which~$\alpha \preccurlyeq \alpha_0$.}
\label{fig:pretty-things}
\end{figure*}

\subsection{Threshold conditions}
\label{subsec-thre}

Next we exhibit a threshold $B$ (defined in terms of the
recurrence~\eqref{eq:rel}) such that for all $n>B$ there are
rational integers $\theta'$ and $p$, with $p>n$ prime, satisfying the
following conditions:
\begin{enumerate}
\item[(P1)] $m_\theta(\theta')\equiv 0\pmod{p}$;
\item[(P2)]  The function $\varphi:\mathcal{R}\rightarrow\Z$ defined by
\[  \varphi(a\theta+b) = \begin{cases}
a\theta' + b & \text{if } a> 0,\\
a\theta' + b + p  & \text{if } a\leq 0
\end{cases}\]
is an order embedding of $(\mathcal{R},\prec)$
in $(\{0,1,\ldots,p-1\},<)$.
\item[(P3)] The set $\{ \alpha \in \mathcal{R} : 1 \leq \varphi(\alpha)
  \leq n \}$ is unbalanced.
\end{enumerate}

The definitions for $\theta'$ and $p$  follow.
Consider
the interval
\begin{gather} I(n):= \left\{ k\in\N : \frac{6n}{3a_0+2}\leq k-1\leq
    \frac{6n}{3a_0+1} \right\} \, .
    \label{eq:INT}
    \end{gather} 
    and let
    $M$ be an upper bound on $\{|a|,|b|: a \theta+b\in 
\mathcal{R}\}$, and the height of the minimal polynomials of the elements of $\mathcal{R}$.
By
\autoref{theo:bigprimespolyT}, there is an effective threshold $B$, which we may assume to be greater than $3M(M+1)$, 
such that for all $n>B$ there exists a prime $p > 3Mn$ that divides the
product
\[ \prod_{\substack{k\in I(n)\\ k \in 2\N+1}} k^2-\beta.\]
Furthermore, since \(p\) is prime, we deduce that there exists 
$k_0 \in I(n) \cap (2\N+1)$ such that $k_0^2 \equiv \beta \pmod{p}$. 
We define $\theta' \in \N$ to be the number such that $k_0=2\theta'+1$.

We will show that $\theta'$ and $p$ satisfy Conditions (P1)--(P3).  Now
\begin{equation*}
    m_\theta(\theta') = m_\theta \biggl(\frac{k_0 - 1}{2} \biggr) \equiv k_0^2-\beta \equiv 0 \pmod{p}.
\end{equation*}
Thus $\theta'$ satisfies
Condition (P1).

We turn next to establishing Condition (P2).
Since $k_0\in I(n)$ and $k_0=2\theta'+1$, we have
\begin{equation}
(a_0+\textstyle\frac{1}{3})\theta' \leq n \leq
(a_0+\textstyle\frac{2}{3})\theta'.
  \label{eq:INEQ}
\end{equation}
Combining~\eqref{eq:INEQ} with the inequality $1\leq a_0 \leq M$ and rearranging terms gives 
$\frac{n}{M+1} \leq \theta' \leq \frac{3n}{4}$.
Recalling that $p>3Mn$ and $n > B \geq 3M(M+1)$, we conclude that
\begin{equation}
3M \leq \theta' \leq \frac{p}{4M} \, . 
\label{eq:INEQQ}
\end{equation}
The inequality $\theta'\leq\frac{p}{4M}$
in~\eqref{eq:INEQQ}
implies that for all roots \(a\theta + b \in \mathcal{R}\), $\varphi(a\theta+b)$ is equal to
    \begin{align*}
       a\theta' + b \in &
            \left\{0,\ldots,\frac{p-1}{2}\right\} \text{ if } a>0, \text{ and } \\
       a\theta' + b + p \in & \left\{\frac{p-1}{2},\ldots,p\right\} \text{ if } a\leq 0
    \end{align*}
(for the latter, recall that $\mathcal{R}$ contains no positive integers).  
Further, since  $|b| \leq M<\theta'$
for all $a\theta+b\in \mathcal{R}$, we conclude that
$\varphi$ is an order
embedding of $(\mathcal{R},\prec)$ into 
$(\{0,\ldots,p-1\},<)$.
This establishes~(P2).

Equation~\eqref{eq:INEQ}  
and the inequality $\theta' \leq  3M$ from~\eqref{eq:INEQQ}
yields \[ \varphi(\alpha_0) < a_0\theta'+M < n < (a_0+1)\theta'-M \, .\] 
Hence \(\varphi(\alpha_0)\), the image of the greatest element in \(\mathcal{C}_0\) is upper bounded by \(n\).
From the definition of the order \((\mathcal{R},\preccurlyeq)\), for
$\alpha\in\mathcal R$ we have that $\alpha\preccurlyeq \alpha_0$ if and only if
$\varphi(\alpha) \leq n$.  
Thus (P3) follows from the fact that the set
$\{\alpha\in\mathcal{R}:\alpha \preccurlyeq \alpha_0\}$ is unbalanced.

\subsection{Prime divisors of \texorpdfstring{$u_n$}{un}}
\label{subsec-primediv}

To conclude the proof, we now explain why properties (P1)--(P3) imply that
$p$ divides $u_n$.  
Define $\psi:\Z[\theta]\rightarrow \Z/p\Z$ by
\[ \psi(a\theta+b):=(a\theta'+b)\bmod{p}.\]  
Condition (P2) entails that
$\psi$ and $\varphi$ agree on $\mathcal R$, while Condition (P1) entails
that $\psi$ is a ring homomorphism. 
(We note in passing that the kernel of $\psi$ is a prime 
ideal~$\mathfrak{p}$ appearing in prime ideal factorisation of $p\Z[\theta]$.)
Hence the polynomial $fg$ splits
over $\Z/p\Z$ and $\varphi$ maps the roots of $fg$ in $\K$ to roots of
$fg$ in $\Z/p\Z$.

Consider the decomposition of the \(p\)-adic valuation
\begin{equation*}
    v_p(u_n) = \sum_{k=1}^n (v_p(g(k))-v_p(f(k))).
\end{equation*}

Let $h(x)$ be an irreducible factor of either $f$ or $g$.
Then $h(x)$ is monic, of degree at most $2$ and height at most $M$.  Since $p>3Mn$, we easily see that $|h(k)|<p^2$
for all $1\leq k \leq n$ and hence
$v_p(h(k)) \in \{0,1\}$.
It follows that $v_p(u_n)$ is equal to the number of roots of $g$ 
in $\Z/p\Z$ that lie in $\{1,\ldots,n\}$ minus the number of roots of $f$ in $\Z/p\Z$ that lie in $\{1,\ldots,n\}$, counting repeated roots according to their multiplicity. 
Observe that this count takes place on the set \(\{\alpha\in\mathcal{R} : 1\le \varphi(\alpha)\le n\}\).
By Condition (P3), the aforementioned set is unbalanced and so it quickly follows that
$v_p(u_n)\neq 0$.

\subsection{Concluding the proof of \autoref{theo-decide-quad}}
\label{subsec-conclud}
Finally, let us return to the decidability of the Membership Problem in the setting of \autoref{theo-decide-quad}.
By our standing assumption that all instances of the problem are normalised we have that $t\neq 0$.  We have exhibited a bound $B$ such that for all $n>B$ there exists a prime $p>3Mn$
such that $v_p(u_n)\neq 0$.  This means that if $p_0$ is the largest prime such that $v_{p_0}(t)\neq 0$ then
for $n>\max\left(B,\frac{p_0}{3M}\right)$ we have $u_n\neq t$.
Thus we have reduced the Membership Problem in this setting to a finite search problem.
This immediately establishes decidability and concludes our proof of \autoref{theo-decide-quad}.

We illustrate the construction underlying   \autoref{theo-decide-quad} with a worked example.

\begin{example} \label{ex:example2}
Let $\langle u_n \rangle_{n=0}^\infty$ be the hypergeometric sequence defined by the recurrence
\begin{equation*}
f(n)u_{n} - g(n)u_{n-1} = 0 \quad \text{with} \quad u_0=1,
 \end{equation*}
where
 \(f(x) :=  x^2 -x -1\) and \(g(x) :=  x^2 + 2x - 4\).

The polynomials  $f$ and $g$ both have
splitting field $\K = \Q(\sqrt{5})$, with  ring of integers $\Z[\frac{\sqrt{5} - 1}{2}]$. Define $\theta := \frac{\sqrt{5} - 1}{2}$, and write $m_\theta(x) = x^2 + x - 1$ for its minimal polynomial.

Since $f=(x-\theta-1)(x+\theta)$ and $g=(x-2\theta)(x+2\theta+2)$,
the set of roots of $fg$ is
$\mathcal{R} = \{\theta + 1, -\theta,  2\theta, -2\theta-2\}$. The definition of the linear ordering $\prec$ on $\mathcal{R}$ 
 (see Section~\ref{subsec-partition-roots}) yields
\[ \theta + 1 \prec 2\theta \prec -2\theta-2 \prec -\theta \, ,\]
with the least unbalanced class being $\mathcal{C}_0 := \{\theta + 1\}$.
Define $M := 4$, which is an upper bound on $\{|a|,|b|: a \theta+b\in 
\mathcal{R}\}$ and the heights of $f$ and $g$ (which are the respective minimal polynomials of the elements of $\mathcal{R}$). 

Write $p_0$ for the largest prime such that $v_{p_0}(t) \neq 0$.
 By \Cref{theo:bigprimespolyT}, there is a bound $B> 3M(M+1)$ such that for all $n > \max(B, \frac{p_0}{3M})$,  there is a prime $p$ with $v_p(u_n) \neq v_{p}(t)$. 
This permits us to reduce  the Membership Problem for \(\langle {u_n}\rangle_{n=0}^{\infty}\)  and $t$ to a finite search problem. 

Given a target $t$ and sufficiently large $n$, the process in the proof of \Cref{theo-decide-quad} 
finds a prime $p$ with 
$v_p(u_n)\neq v_p(t)$.  Below we illustrate the idea of the proof
in the specific case 
$t=\frac{11}{59}$ and $n=61$.  (Here we have $p_0=59$ and hence  $n>3M(M+1)$ and $n>\frac{p_0}{3M}$, as required in the proof of Theorem~\ref{theo-decide-quad}.) We will establish the existence of a  
 prime $p$ such that 
\(v_p(u_{61})\neq 0\) \text{and}  \(v_p(t)=0\),
witnessing that $u_{61} \neq t$.

Guided by the proof of  \Cref{theo:bigprimespolyT}, we observe that prime $p:=1481 > 3n M$ is a divisor of 
\[ \prod_{\substack{k\in I(61)\\ k \in 2\N+1}} k^2-\beta=(75^2 - 5)(77^2 - 5) \cdots(91^2-5).\]
In particular, we have $p|(77^2-5)$.
Choosing $\theta':=\frac{77-1}{2}=38$, we observe that
the pair  $p$ and $\theta'
$ satisfy conditions (P1)-(P3) in Section~\ref{subsec-thre}:
\begin{enumerate}
   \item[(P1)] $m_\theta(\theta')=38^2+38-1\equiv 0\pmod{p}$;
   \item[(P2)] 
The map 
$\varphi:\mathcal{R}\rightarrow \Z/p\Z$  is an order embedding of $(\mathcal{R},\prec)$
into $(\{0,\ldots,p-1\},<)$,  which can be seen  by noting that
\begin{align*}
  \varphi(\theta + 1) = 39 \qquad & \qquad  \varphi(2\theta) = 76\\
  \varphi(-2\theta -2) = 1403 \qquad & \qquad  \varphi(-\theta) = 1443 ,
\end{align*}
whence $\varphi(\theta + 1)  <  \varphi(2\theta) <  \varphi(-2\theta -2) <  \varphi(-\theta)$.
 \item[(P3)] The set $\{\alpha \in \mathcal{R}:1\leq \varphi(\alpha) \leq 61\} = \{\theta+1\}$ is unbalanced.
\end{enumerate}

By the  arguments above, in the equation
\[ v_p(u_{61}) = \sum_{k=1}^{61} (v_p(g(k))-v_p(f(k))),\]
the only non-zero term on the right-hand side is
\[ v_p(f(\varphi(\theta + 1))) = v_p(f(39)) = v_p(1481) = 1.\]
It follows that $v_p(u_{61}) = -1$, while
$v_p(t)=0$.

\end{example}

\section{Discussion} \label{sec:discussion}
In light of the results in~\Cref{sec-unequalsplitting}
a clear direction for further research is to 
examine the decidability of the Membership Problem for recurrences whose polynomial coefficients share the same splitting field.  
We recall that previous work \cite{NPSW022} established decidability when the polynomial coefficients split over the rationals.  The present work considers the case when the two polynomials split over the ring of integers of a quadratic field.  In future work we will consider the more general case in which the all roots of the coefficient polynomials have degree at most two. 
As far as the authors are aware, the only known results in this direction are the (un)conditional decidability results for quadratic parameters in~\cite{kenison2022applications}.
Extending the approach of the present paper to the case of polynomials with roots of degree more than two would require new results on large prime divisors on the values of such polynomials, which is an active area of research in number theory.

%\SkipTocEntry\section*{Acknowledgements}

\bibliographystyle{plain}
\bibliography{literature}

\begin{thebibliography}{10}

\bibitem{Moll09}
T.~Amdeberhan, L.~Medina, and V.~Moll.
\newblock Asymptotic valuations of sequences satisfying first order
  recurrences.
\newblock {\em Proceedings of the American Mathematical Society},
  137(3):885--890, 2009.

\bibitem{apostol1998introduction}
T.~M. Apostol.
\newblock {\em Introduction to Analytic Number Theory}.
\newblock Undergraduate Texts in Mathematics. Springer New York, 1998.

\bibitem{everest07}
G.~Everest, S.~Stevens, D.~Tamsett, and T.~Ward.
\newblock Primes generated by recurrence sequences.
\newblock {\em The American Mathematical Monthly}, 114(5):417--431, 2007.

\bibitem{everest2003recurrence}
G.~Everest, A.~van~der Poorten, I.~Shparlinski, and T.~Ward.
\newblock {\em Recurrence sequences}, volume 104 of {\em Math. Surveys Monogr.}
\newblock Amer. Math. Soc., Providence, RI, 2003.

\bibitem{FS09}
P.~Flajolet and R.~Sedgewick.
\newblock {\em Analytic Combinatorics}.
\newblock Cambridge University Press, 2009.

\bibitem{gouvea2020padic}
F.~Q. Gouv\^{e}a.
\newblock {\em {$p$}-adic numbers}.
\newblock Universitext. Springer, Cham, 3 edition, 2020.

\bibitem{heathbrown2001largest}
D.~Heath-Brown.
\newblock The largest prime factor of $x^3+2$.
\newblock {\em Proceedings of the London Mathematical Society}, 82(3):554--596,
  2001.

\bibitem{hinz1996multiplicative}
J.~G. Hinz.
\newblock A generalization of a problem of chebyshev.
\newblock {\em Acta Arithmetica}, 74(3):207--230, 1996.

\bibitem{hongarxiv2016}
S.~Hong and C.~Wang.
\newblock Criterion for the integrality of hypergeometric series with
  parameters from quadratic fields, 2016.

\bibitem{Concrete11}
M.~Kauers and P.~Paule.
\newblock {\em The Concrete Tetrahedron}.
\newblock Springer Vienna, 2011.

\bibitem{KauersP10}
M.~Kauers and V.~Pillwein.
\newblock When can we detect that a p-finite sequence is positive?
\newblock In {\em Symbolic and Algebraic Computation, International Symposium,
  {ISSAC}, Proceedings}, pages 195--201. {ACM}, 2010.

\bibitem{kenison2022applications}
G.~Kenison.
\newblock A transcendental approach to decision problems for hypergeometric
  sequences.
\newblock Submitted.

\bibitem{kenison2020positivity}
G.~Kenison, O.~Klurman, E.~Lefaucheux, F.~Luca, P.~Moree, J.~Ouaknine,
  A.~Whiteland, and J.~Worrell.
\newblock On inequality decision problems for low-order holonomic sequences.
\newblock Submitted.

\bibitem{KenisonKLLMOW021}
G.~Kenison, O.~Klurman, E.~Lefaucheux, F.~Luca, P.~Moree, J.~Ouaknine, M.~A.
  Whiteland, and J.~Worrell.
\newblock On positivity and minimality for second-order holonomic sequences.
\newblock In {\em 46th International Symposium on Mathematical Foundations of
  Computer Science, {MFCS}}, volume 202 of {\em LIPIcs}, pages 67:1--67:15.
  Schloss Dagstuhl - Leibniz-Zentrum f{\"{u}}r Informatik, 2021.

\bibitem{landau1900factorielles}
E.~Landau.
\newblock Sur les conditions de divisibilit\'e d'un produit de factorielles par
  un autre.
\newblock {\em Nouvelles annales de math\'ematiques : journal des candidats aux
  \'ecoles polytechnique et normale}, 3e s{\'e}rie, 19:344--362, 1900.

\bibitem{mertens1874zahlentheorie}
F.~Mertens.
\newblock Ein {B}eitrag zur analytischen {Z}ahlentheorie.
\newblock {\em J. Reine Angew. Math.}, 78:46--62, 1874.

\bibitem{mignotte1984distance}
M.~Mignotte, T.~Shorey, and R.~Tijdeman.
\newblock The distance between terms of an algebraic recurrence sequence.
\newblock {\em Journal f\"{u}r die Reine und Angewandte Mathematik}, pages
  63--76, 1984.

\bibitem{milneANT}
J.~S. Milne.
\newblock Algebraic number theory (v3.08), 2020.
\newblock Available at \url{www.jmilne.org/math/}.

\bibitem{nesterenko1996modular}
Y.~Nesterenko.
\newblock Modular functions and transcendence problems.
\newblock {\em C. R. Acad. Sci. Paris S\'{e}r. I Math.}, 322(10):909--914,
  1996.

\bibitem{neukirch1999algebraic}
J.~Neukirch.
\newblock {\em Algebraic number theory}, volume 322 of {\em Grundlehren der
  mathematischen Wissenschaften [Fundamental Principles of Mathematical
  Sciences]}.
\newblock Springer-Verlag, Berlin, 1999.
\newblock Translated from the 1992 German original and with a note by Norbert
  Schappacher, With a foreword by G. Harder.

\bibitem{NeumannO021}
E.~Neumann, J.~Ouaknine, and J.~Worrell.
\newblock Decision problems for second-order holonomic recurrences.
\newblock In {\em 48th International Colloquium on Automata, Languages, and
  Programming, {ICALP}}, volume 198 of {\em LIPIcs}, pages 99:1--99:20. Schloss
  Dagstuhl - Leibniz-Zentrum f{\"{u}}r Informatik, 2021.

\bibitem{NPSW022}
K.~Nosan, A.~Pouly, M.~Shirmohammadi, and J.~Worrell.
\newblock The membership problem for hypergeometric sequences with rational
  parameters.
\newblock In {\em Proceedings of the 2022 International Symposium on Symbolic
  and Algebraic Computation}, ISSAC '22, page 381–389, New York, NY, USA,
  2022. Association for Computing Machinery.

\bibitem{PillweinS15}
V.~Pillwein and M.~Schussler.
\newblock An efficient procedure deciding positivity for a class of holonomic
  functions.
\newblock {\em {ACM} Commun. Comput. Algebra}, 49(3):90--93, 2015.

\bibitem{selberg1950pnt-ap}
A.~Selberg.
\newblock An elementary proof of the prime-number theorem for arithmetic
  progressions.
\newblock {\em Canadian Journal of Mathematics}, 2:66–78, 1950.

\bibitem{stewart2016algebraic}
I.~Stewart and D.~Tall.
\newblock {\em Algebraic number theory and {F}ermat's last theorem}.
\newblock CRC Press, Boca Raton, FL, fourth edition, 2016.

\bibitem{vereshchagin1985occurence}
N.~Vereshchagin.
\newblock Occurrence of zero in a linear recursive sequence.
\newblock {\em Mathematical notes of the Academy of Sciences of the USSR},
  38(2):609--615, Aug 1985.

\bibitem{wallis1655arithmetica}
J.~Wallis.
\newblock Arithmetica infinitorum, sive nova methodus inquirendi in
  curvilineorum quadraturam, aliaque difficiliori matheseos problemata.
\newblock {\em Oxford}, pages 1--199, 1655.

\end{thebibliography}

\appendix

\section{Proofs for Section~\ref{sec-polyseq}}
\label{sec:app-lowerbownun}

\begin{proof}[Proof of \autoref{claim:sizefunnyQ}]
First note that $A=\frac{b-a}{c}$.
The claim states that
\[ \log(F_n)	\geq \frac{2(b-a)}{c} (n\log n - n) \, .\]

The proof is as follows.
Given $y\in \N$, we first observe that 
\[\prod_{cx \leq y} (cx)^2 \geq   c^{2y}\bigg(\bigg\lfloor\frac{y}{c} \bigg\rfloor!\bigg)^2.\]
By Stirling's formula, the logarithm of the  quantity above is at least 
\begin{equation}
\label{eq:logQI}
\frac{2y}{c}\log c+ \frac{2y}{c} \log y -\frac{2y}{c}.
\end{equation}

Now \(F_n = \prod_{k\in I(n)} (k^2+\beta)\) is bounded from below by
\begin{equation*}
	F_n \geq \prod_{k\in I(n)} k^2
	\geq \prod_{ an\leq cx \leq bn} (cx+d)^2 
	\geq \prod_{ an\leq cx \leq bn} (cx)^2.
	\end{equation*}
By the above, and Equation~\eqref{eq:logQI} we conclude that \(\log(F_n)\) is bounded from below by
\begin{equation*}
	\log \prod_{cx \leq bn} c^2x^2  - \log \prod_{cx \leq an} c^2x^2
	\geq  \frac{2(b-a)}{c}( n \log n -n),
\end{equation*}
as required.  
\end{proof}

We now prove the inequality~\eqref{eq:BOUND} from the proof
of~\Cref{theo:bigprimespolyT}.
Noting that $A=\frac{b-a}{c}$, the inequality states
that
\begin{equation}
  e_p \leq  \frac{2An}{p-1}+  \frac{\varepsilon_2 \log n}{\log p}
\tag{\ref{eq:BOUND}}
\end{equation}

\begin{proof}[Proof of Inequality \eqref{eq:BOUND}]

If $e_p=0$ then the bound trivially holds.
Suppose $e_p>0$.  Then the function $f$ has two roots in $\Z/p\Z$.
Define $g\in\Z[x]$ by $g(x):=f(cx+d)$.
In case $p>c$ then $g$ also has two roots in $\Z/p\Z$.
For all $n\in\N$ define the products
\begin{equation*}
    G_n:=\prod_{k=1}^{\left\lfloor \frac{bn-d}{c} \right\rfloor} g(k)
    \quad \text{and} \quad
    H_n:=\prod_{k=1}^{\left\lceil \frac{an-d}{c} \right\rceil - 1} g(k)
\end{equation*}
    Then $F_n=\frac{G_n}{H_n}$ and hence
    $e_p=v_p(F_n)=v_p(G_n)-v_p(H_n)$.
Applying \Cref{prop:moll}, we get, for some constant $\varepsilon>0$,
\begin{align*}
v_p(G_n) \leq & \frac{2(bn-d)}{c(p-1)} + \frac{\varepsilon\log 
n}{\log p} \qquad {\text{ and }  }\\
v_p(H_n) \geq & \frac{2(an-d-c)}{c(p-1)} -
\frac{\varepsilon\log 
n}{\log p}.
\end{align*}
The upper bound in~\eqref{eq:BOUND} follows, for a suitable
choice of the constant $\varepsilon_2$, by subtracting the
upper bound for $v_p(G_n)$ from the lower bound for $v_p(H_n)$.
\end{proof}

\section{Second Case in the proof of Theorem~\ref{theo-decide-quad}}
\label{app-sec-quad}

Let $\beta \not\equiv 1 \pmod{4}$ be a square-free integer and
$\K=\Q(\sqrt{\beta})$ a quadratic field over which the polynomials $f$
and $g$ in~\eqref{eq:rel} split completely.  By Theorem~\ref{thm:quadratic}, the ring of integers of the field 
$\K$ is $\Z[\sqrt{\beta}]$. We define $\theta:=\sqrt{\beta}$,
so that $m_{\theta}:=x^2-\beta$ is the minimal polynomial of~$\theta$. 

Exactly as in Subsection~\ref{subsec-partition-roots}, we partition the set $\mathcal R$ of roots of $fg$ into classes, define the balanced and unbalanced classes, define the linear ordering~$\prec$ on $\mathcal R$, and consider the \emph{least unbalanced class} $C_0$.  
Let $a_0\theta+b_0$ be the greatest element in $C_0$ and note that $a_0\geq 1$ as before.

\subsection{Threshold conditions}
Next we exhibit a threshold $B$ (defined in terms of the
recurrence~\eqref{eq:rel}) such that for all $n>B$ there are
rational integers $\theta'$ and $p$, with $p>n$ prime, satisfying the three
conditions (P1)--(P3) as stated in Subsection~\ref{subsec-thre}.

The definitions for $\theta'$ and $p$ are as follows.
Consider
the interval
\[ I(n):= \left\{ k\in\N : \frac{3n}{3a_0+2}\leq k \leq
    \frac{3n}{3a_0+1} \right\} \] 
    and let
    $M$ be an upper bound on $\{|a|,|b|: a \theta+b\in 
\mathcal{R}\}$, and the height of the minimal polynomials of the elements of $\mathcal{R}$.
By
\autoref{theo:bigprimespolyT}, there is an effective threshold $B$, which we may assume to be greater than $3M(M+1)$,
such that for all $n>B$ there exists a prime $p > 3Mn$ that divides the
product
\[ \prod_{k\in I(n)} k^2-\beta.\]
Further, since \(p>3Mn\) is prime, we deduce that for \(n>B\) there exists 
$\theta' \in I(n)$ such that $(\theta')^2 \equiv \beta \pmod{p}$. 

We will show that $\theta'$ and $p$ satisfy Conditions (P1)--(P3).  Now
\begin{equation*}
    m_\theta(\theta') \equiv (\theta')^2-\beta \equiv 0 \pmod{p}.
\end{equation*}
Thus $\theta'$ satisfies
Condition (P1).

We turn next to establishing Condition (P2).
Since $\theta'\in I(n)$, it is straightforward to show that
\begin{equation}
(a_0+\textstyle\frac{1}{3})\theta' \leq n \leq
(a_0+\textstyle\frac{2}{3})\theta' \, .
  \label{app-eq:INEQ}
\end{equation}
These 
  bounds are identical to those in~\eqref{eq:INEQ}.
  In this case, Conditions (P2) and (P3) follow by an analogous argument to that given in Subsection~\ref{subsec-thre}.

\bigskip
\bigskip

The remaining part of the proof for the case $\beta\equiv  1 \pmod{4}$, as given  in Subsection~\ref{subsec-primediv} and Subsection~~\ref{subsec-conclud}, carries over to the present case without change.

\end{document}